\documentclass[11pt,letterpaper,oneside,english]{article}
\usepackage[top=3.2cm, bottom=3.2cm, left=3.2cm, right=3.2cm]{geometry}

\usepackage{mdwlist}

\usepackage{amssymb,amsbsy,latexsym}
\usepackage{amsmath}
\usepackage{graphics, subfigure, float}
\usepackage{fp, calc}
\usepackage{bm}

\usepackage[latin1]{inputenc}
\usepackage[american]{babel}
\usepackage[T1]{fontenc} 

\def\url#1{{\texttt #1}}

\usepackage{bm}

\usepackage{amscd,amsthm}
\usepackage{datetime}

\usepackage{pst-all}
\usepackage{verbatim, comment}

\newtheoremstyle{theorem}{1em}{1em}{\slshape}{0pt}{\bfseries}{.}{ }{}
\theoremstyle{theorem}
\newtheorem{theorem}{Theorem}
\newtheorem*{theorem*}{Theorem}

\newtheorem{lemma}[theorem]{Lemma}

\theoremstyle{remark}

\newtheorem*{remark*}{Remark}

\providecommand{\setR}{\mathbb{R}}

\newcommand{\conv}{\textrm{conv}}
\newcommand{\rank}{\textrm{rank}}

\newcommand{\E}{\mathop{\mathbb{E}}}

 \theoremstyle{theorem}
\newtheorem*{homework*}{Homework}
 \theoremstyle{definition}
 
\newtheorem*{claim*}{Claim}

\newtheorem*{example*}{Beispiel}

\usepackage[displaymath,textmath,graphics, subfigure, floats]{preview} 
\PreviewEnvironment{myproblem}
\PreviewEnvironment{defn} 
\PreviewEnvironment{center} 
\PreviewEnvironment{pspicture} 
 \usepackage{datetime}

\makeatother

\begin{document}

\title{A direct proof for Lovett's bound on the communication complexity of low rank matrices}
\date{} 
\author{Thomas Rothvo{\ss}\thanks{Email: {\tt rothvoss@uw.edu}. Supported by NSF grant 1420180 with title ``\emph{Limitations of convex relaxations in combinatorial optimization}''. } 
\vspace{2mm} \\ University of Washington, Seattle} 
\maketitle

\begin{abstract}
The \emph{log-rank conjecture} in communication complexity suggests that the deterministic communication
complexity of any Boolean rank-$r$ function is bounded by $\textrm{polylog}(r)$. 
Recently, major progress was made by Lovett who proved that the communication complexity is bounded
by $O(\sqrt{r} \cdot \log r)$.
Lovett's proof is based on known estimates on the discrepancy of low-rank matrices. We give a simple, 
direct proof based on a hyperplane rounding argument that in our opinion sheds more light on the reason 
why a root factor suffices and what is necessary to improve on this factor. 
\end{abstract}

\section{Introduction}

In the classical \emph{communication complexity} setting, we imagine to have two players, Alice and Bob
and a function $f : X \times Y \to \{ \pm 1\}$. The players agree on a communication protocol beforehand; then
Alice is given an input $x \in X$ and Bob is presented an input $y \in Y$. Then the players can exchange messages
to figure out the function value $f(x,y)$ of their common input. The cost of the protocol is the number 
of exchanged bits for the worst case input. Moreover we denote the cost of the most efficient protocol by $CC^{\det}(f)$.

It is common to view the function $f$ as a matrix $M \in \{ \pm 1\}^{X \times Y}$ with entries $M_{xy} = f(x,y)$ --- we will interchangeably use the function $f$ and the matrix $M$ and we abbreviate $\rank(f) := \rank(M)$.
A \emph{monochromatic rectangle} for $f$ is a subset $R = X' \times Y'$ with $X' \subseteq X$ and $Y' \subseteq Y$ on which the function is constant. In particular, the leaves of the optimal deterministic protocol tree correspond to a partition
of $M$ into $2^{CC^{\textrm{det}}(f)}$ many monochromatic rectangles. 
Observe that this partition can be used to write $M$ as the sum of $2^{CC^{\textrm{det}}(f)}$ many rank-1 matrices, 
which implies that 
$CC^{\textrm{det}}(f) \geq \log \rank(f)$. On the other hand it is also known that $CC^{\textrm{det}}(f) \leq \rank(f)$. 
In fact, Lov{\'a}sz and Saks~\cite{LatticesMoebiusCommComplexity-LovaszSaks-FOCS88} even conjectured that 
the rank lower bound is tight up to a polynomial factor, that means  $CC^{\textrm{det}}(f) \leq (\log \textrm{rank}(r))^{O(1)}$.
The exponent in this \emph{log-rank conjecture} needs to be at least $\log_3(6) \approx 1.63$ (unpublished by Kushilevitz, cf. \cite{RankVsCommunicationComplexity-NisanWigderson95}). 
Small improvements have been made by Kotlov~\cite{Rank-and-chromatic-number-of-a-graph-Kotlov97}, 
who showed that $CC^{\textrm{det}}(f) \leq \log(4/3) \rank(f)$
and Ben-Sasson, Ron-Zewi and Lovett~\cite{AddCombinatorics-approach-comm-complexity-BenSassonLovettRonZewi-FOCS12} 
who gave an asymptotic improvement of $CC^{\textrm{det}}(f) \leq O(\frac{\rank(f)}{\log \rank(f)})$, 
but had to assume the polynomial Freiman Rusza conjecture. 

In a recent breakthrough, Lovett~\cite{CommunicationBoundedByRootRank-Lovett-STOC2014} showed 
an unconditional bound of $CC^{\textrm{det}}(f) \leq O(\sqrt{r} \log r)$. 
The key ingredient for his result is a lower bound on the \emph{discrepancy} of a function/matrix.

\begin{theorem}[\cite{CommunicationComplexity-KushilevitzNisan1995,Complexity-measures-of-sign-matrices-LMSS-Combinatorica2007}]
For any rank-$r$ Boolean matrix $M$ and any measure $\mu$ on its entries, 
there exists a rectangle $R$ so that
$
| \mu(f^{-1}(1) \cap R) - \mu(f^{-1}(-1) \cap R)| \geq \frac{1}{8\sqrt{r}}
$.
\end{theorem}
Formally, the discrepancy of a function is the minimum such quantity over all possible measures,
\[
  \textrm{disc}(f) = \min_{\textrm{measure }\mu} \; \max_{\textrm{rectangle } R} | \mu(f^{-1}(1) \cap R) - \mu(f^{-1}(-1) \cap R)|
\]
The discrepancy lower bound is tight in general, that means there are indeed functions $f$ with
$\textrm{disc}(f) \leq O(\frac{1}{\sqrt{\rank(f)}})$, which suggests that using discrepancy as a black box might not be enough for
a better bound. Recently Shraibman~\cite{CorruptionBoundCommComplexityShraibman-ARXIV2014} 
obtained  
Lovett's bound in terms of the corruption lower bound. However, for expressing the bound in 
terms of the rank of the matrix, the result still relies on the black-box bound on the discrepancy. 

In this write-up, we give a direct proof that bypasses the discrepancy lower bound and
somewhat makes it more clear, what needs to be done in order to break the $\sqrt{\rank(f)}$ barrier. 
For more details on communication complexity we refer to the book of Kushilevitz and 
Nisan~\cite{CommunicationComplexity-KushilevitzNisan1995}.

\section{Preliminaries}

The main technical result of Lovett~\cite{CommunicationBoundedByRootRank-Lovett-STOC2014} in a 
slightly paraphrased form says that we can always find a large
rectangle $R$ that is \emph{almost monochromatic}. 
\begin{theorem} \label{thm:FindingAlmostMonochromaticRectangle}
Given any Boolean function $f : X \times Y \to  \{ \pm 1\}$ 
with rank $r$ and a measure $\mu$ on $X \times Y$
with $\mu(f^{-1}(1)) \geq \delta>0$. Then there exists a rectangle $R \subseteq X \times Y$ 
with $\mu(R) \geq 2^{-\Theta(\sqrt{r} \log \frac{1}{\delta})}$ so that 
$\E_{(x,y) \sim R}[ f(x,y) ] \geq 1- \delta$.
\end{theorem}
In particular one can use Theorem~\ref{thm:FindingAlmostMonochromaticRectangle} to find a rectangle $R$ of size $|R| \geq 2^{\Theta(\sqrt{r} \log r)}|X \times Y|$ which has a $(1-\frac{1}{8r})$-fraction of 1-entries (assuming for symmetry reasons
that at least half of the entries of $f$ were 1). 
By arguments of Gavinsky and Lovett~\cite{EnRouteToLogRankConj-Gavinsky-Lovett-ICALP2014} such an almost
monochromatic rectangle always contains a sub-rectangle $R' \subseteq R$ with $|R'| \geq \frac{1}{8}|R|$
that is \emph{fully} monochromatic.

The guarantee of having large monochromatic rectangles in any sub-matrix, can then be turned into a 
protocol using arguments of Nisan and Wigderson:
\begin{theorem}[\cite{RankVsCommunicationComplexity-NisanWigderson95}] \label{thm:NisanWigdersonProtocoll}
Assume that any rank-$r$ function $f : X \times Y \to \{ \pm 1\}$ has 
a monochromatic rectangle of size $2^{-c(r)}$. Then any Boolean function $g$
has
\[
  CC^{\textrm{det}}(g) \leq O(\log^2 \rank(g)) + \sum_{i=0}^{\log \rank(g)} O(c(\rank(g)/2^i)).
\]
\end{theorem}
In particular, we will apply Theorem~\ref{thm:NisanWigdersonProtocoll} with $c(r) = \Theta(\sqrt{r} \log(r))$
and obtain a protocol of cost $O(\sqrt{r} \log r)$ for any rank $r$ function. 
We want to emphasize that the whole construction only has a $\textrm{polylog}(r)$
overhead, that means the log-rank conjecture is actually \emph{equivalent} to being able to find rectangles
of size $2^{-\textrm{polylog}(r)}|X \times Y|$ that have at least a $1-\frac{1}{8r}$ fraction of entries $+1$ (or $-1$, resp).
The reader can find a more detailed explanation of Theorem~\ref{thm:NisanWigdersonProtocoll} in Lovett's paper~\cite{CommunicationBoundedByRootRank-Lovett-STOC2014}.



\section{Proof of the main theorem}

In this section, we want to reprove Lovett's main technical result, Theorem~\ref{thm:FindingAlmostMonochromaticRectangle}.
Fix a matrix $M \in \{ \pm 1\}^{X \times Y}$ and denote its rank by $r$. 
First, what does it actually mean that the matrix has rank $r$? By definition it means that there are
$r$-dimensional vectors $u_x,v_y$ for all $x \in X$ and $y \in Y$ so that $\left<u_x,u_y\right> = M_{xy}$. But what can we actually
say about the \emph{length} of those vectors? 
To quote Linial, Mendelson, Schechtman and Shraibman, it is ``well known to Banach space theorists''
that length $r^{1/4}$ suffices
(see \cite{Complexity-measures-of-sign-matrices-LMSS-Combinatorica2007}, Lemma 4.2).
For the case that the Banach space knowledge of the reader got a bit rusty, we include 
a more or less self-contained proof. In the exposition we follow closely \cite{FawziGouveiaParriloRobinsonThomas-PSD-Rank-Arxiv14}.

\begin{lemma}
Any rank-$r$ matrix $M \in \{ \pm 1\}^{X \times Y}$ 
has a factorization $M = \left<u_x,v_y\right>$  
so that $u_x,v_y \in \setR^r$ 
are vectors  with $\|u_x\|_2,\|v_y\|_2 \leq r^{1/4}$ for $x \in X$, $y \in Y$.
\end{lemma}
\begin{proof}
First of all, by the definition of rank, there are \emph{some} vectors $u_x,v_y \in \setR^r$ so that 
$\left<u_x,v_y\right> = M_{xy}$ with $\textrm{span}\{ u_x : x \in X\} = \setR^r$ --- just that we have no a priori guarantee on their length. 
Observe that this choice of vectors is far from being unique. For example
we could choose any regular matrix $T \in \setR^{r \times r}$ and rescale
 $u_x' := Tu_x$ and $v_y'=(T^{-1})^Tv_y$. The inner product would remain invariant
as $\big<u_x',v_y'\big> = u_x^TT^T(T^{-1})^Tv_y = u_x^Tv_y = M_{xy}$.

To find a suitable linear map $T$, we will make use of \emph{John's Theorem} (\cite{JohnsTheorem1948}, see also the excellent survey of \cite{IntroToModernConvexGeometry-Ball97}):
\begin{theorem}[John '48]
For any full-dimensional symmetric convex set $K \subseteq \setR^r$ and any Ellipsoid $E \subseteq \setR^r$
that is centered at the origin, there exists an invertible linear map $T$ so that $E \subseteq T(K) \subseteq \sqrt{r} E$.
\end{theorem}
We want to apply John's Theorem to $K = \conv\{ \pm u_x \mid x \in X\}$ (which indeed is a symmetric convex set)
and the ellipsoid $E := r^{-1/4}B$ with $B := \{ x \in \setR^r \mid \|x\|_2 = 1\}$ being the unit ball.
First, John's Theorem provides us with a linear map $T$ so that  $r^{-1/4}B \subseteq \textrm{conv}\{ \pm Tu_x : x \in X\} \subseteq r^{1/4}B$. 
Now, we can rescale the vectors by letting $u_x' := Tu_x$ and $v_y' := (T^{-1})^Tv_y$.
For the sake of a simpler notation, let us start all over and assume that the original vectors $u_x$ and $v_y$
satisfied $r^{-1/4}B \subseteq K \subseteq r^{1/4}B$ for $K = \textrm{conv}\{ \pm u_x \mid x \in X\}$ from the beginning on.

Then by this assumption we immediately see that $\|u_x\|_2 \leq r^{1/4}$
and it just remains to argue that also  $\|v_y\|_2 \leq r^{1/4}$ for a fixed $y \in Y$.
To see this, take the vector $w := \frac{v_y}{r^{1/4}\|v_y\|_2}$ and observe that $w \in r^{-1/4}B$
and hence $w \in K$. By standard linear optimization reasoning, there must be a \emph{vertex} 
$\pm u_x$ of $K$ so that $\left|\left<u_x,v_y\right>\right| \geq \left|\left<w,v_y\right>\right|$. 
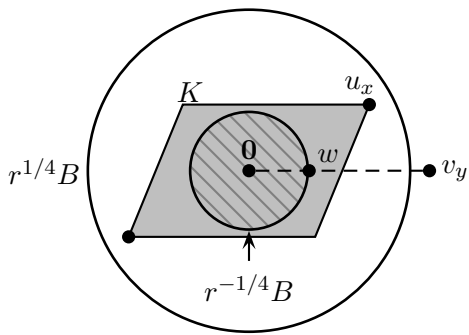
\begin{center}
\psset{unit=0.8cm}
\begin{pspicture}(-4,-2.0)(4,2.5)
\pscircle[linewidth=1pt](0,0){2.7}
\pspolygon[fillstyle=solid,fillcolor=lightgray](2,1.1)(1.1,-1.1)(-2,-1.1)(-1.1,1.1)
\pscircle[linewidth=1pt,fillstyle=vlines,hatchcolor=gray](0,0){1}
\cnode*(0,0){2.5pt}{origin} \nput[labelsep=2pt]{90}{origin}{$\bm{0}$}
\cnode*(2,1.1){2.5pt}{u}
\cnode*(-2,-1.1){2.5pt}{u2}
\nput[labelsep=2pt]{110}{u}{$u_x$}
\rput[l](-1.2,1.3){$K$} 
\cnode*(3,0){2.5pt}{v} \nput[labelsep=2pt]{0}{v}{$v_y$}
\ncline[linestyle=dashed]{origin}{v}
\pnode(0,-1){A} \pnode(0,-1.5){B} \ncline[arrowsize=5pt]{->}{B}{A} \nput{-90}{B}{$r^{-1/4}B$}
\rput[r](-2.8,0){$r^{1/4}B$}
\cnode*(1,0){2.5pt}{w} \nput[labelsep=2pt]{45}{w}{$w$}
\end{pspicture}
\end{center}
This implies that
\[
 r^{-1/4}\|v_y\|_2  = \big|\big<w,v_y\big>\big| \leq \big|\big<u_x,v_y\big>\big| = 1
\]
and the claim is proven.
\end{proof}

\subsection*{The hyperplane rounding argument}

Eventually we are ready to prove Lovett's claim. Let $M \in \{ \pm 1\}^{X \times Y}$
be the matrix with rank-$r$ factorization $M_{xy} = \big<u_x,v_y\big>$ so that $\|u_x\|_2,\|v_y\|_2 \leq r^{1/4}$.
We abbreviate $Q_i = \{ (x,y) \in X \times Y : M_{xy} = i\}$ as the $i$-entries of the matrix. 
We assume that we have a measure $\mu$ with $\mu(Q_1) \geq \delta$ with $\delta > 0$ and we will aim at
finding a large rectangle that contains mostly 1-entries.
It will be convenient to normalize the vectors to $\bar{u}_x := \frac{u_x}{\|u_x\|_2}$ and $\bar{v}_y := \frac{v_y}{\|v_y\|_2}$.
We can make the following observation about their inner products:
\[
  \left<\bar{u}_x,\bar{v}_y\right> = \frac{\left<u_x,v_y\right>}{\|u_x\|_2 \cdot \|v_y\|_2}
\; \begin{cases} \geq \frac{1}{\sqrt{r}} & \textrm{if } M_{xy} = 1 \\ \leq -\frac{1}{\sqrt{r}} & \textrm{if } M_{xy} = -1
\end{cases}
\]
In other words, the \emph{angle} between $u_x$ and $v_y$ for a 1-entries $(x,y)$ is a tiny bit smaller than the 
angle for a $-1$-entry. It is a standard argument that has been used many times e.g. in approximation algorithms 
that if we take a \emph{random hyperplane}, then the chance that a pair of vectors ends up on the same side, is larger
if their angle is smaller. 
Formally, let $N^r(0,1)$ be the distribution of an \emph{$r$-dimensional Gaussian random variable}. 
Then in a slightly modified form, \emph{Sheppard's Formula} tells us:
\begin{lemma}
For any unit vectors $u,v \in \setR^r$ with  $\left<u,v\right> = \alpha$ we have
\[
 \Pr_{g \sim N^r(0,1)}[\left<g,u\right> \geq 0\textrm{ and }\left<g,v\right> \geq 0]
 = \frac{1}{2}\Big(1 - \frac{\textrm{arccos}(\alpha)}{\pi}\Big)
\]
\end{lemma}
In particular, the quantity $\frac{1}{2}(1-\frac{1}{\pi} \textrm{arccos}(\alpha))$ is monotonically increasing
in $\alpha$ with $\frac{1}{2}(1-\frac{1}{\pi} \textrm{arccos}(\alpha)) \geq \frac{1}{4}$ for all $\alpha \geq 0$
and $\frac{1}{2}(1-\frac{1}{\pi} \textrm{arccos}(\alpha)) \leq \frac{1}{4} - \frac{|\alpha|}{7}$
for $\alpha \leq 0$.

Next, we want to take $T := 7\ln(\frac{2}{\delta}) \cdot \sqrt{r}$ many random hyperplanes and define $R$ as those vectors $u_x$
and $v_y$ that always ended up on the positive side. Formally, we will take 
independent random Gaussian vectors $g_1,\ldots,g_T \sim N^r(0,1)$ and define  
rectangles 
\[
R_t := \{ x \in X : \left<\bar{u}_x,g_t\right> \geq 0 \} \times \{ y \in Y : \left<\bar{v}_y,g_t\right> \geq 0\}
\]
and $R := R_1 \cap \ldots \cap R_T$. It remains to argue that in expectation $R$ satisfies the claim of Theorem~\ref{thm:FindingAlmostMonochromaticRectangle}. 

First, using Sheppard's Formula, we know that for an entry
$(x,y) \in Q_1$ one has $\Pr[(x,y) \in R_t] \geq \frac{1}{4}$, while for an entry
$(x,y) \in Q_{-1}$ one has $\Pr[(x,y) \in R_t] \leq \frac{1}{4} - \frac{1}{7\sqrt{r}}$.
Since we take the Gaussians independently, 
\[
  \E[\mu(R \cap Q_1)] \geq \mu(Q_1) \cdot \left(\frac{1}{4}\right)^T
\;\textrm{  and  }\quad \E[\mu(R \cap Q_{-1})] \leq \mu(Q_{-1}) \cdot \left(\frac{1}{4} - \frac{1}{7\sqrt{r}}\right)^T
\]
In particular their ratio behaves like
\[
  \frac{\E[\mu(R \cap Q_{-1})]}{\E[\mu(R \cap Q_1)]}
 \leq \frac{ (1/4 - \frac{1}{7\sqrt{r}})^T}{\delta \cdot (1/4)^T }
 = \frac{1}{\delta} \cdot \Big(1-\frac{4}{7\sqrt{r}}\Big)^T \leq \frac{1}{\delta} \exp\Big(- T \cdot \frac{4}{7\sqrt{r} } \Big) 
\leq \frac{\delta}{2}
\]
for our choice of $T = 7\ln(\frac{2}{\delta}) \cdot \sqrt{r}$.
On the other hand, $\E[\mu(R)] \geq \E[\mu(R \cap Q_1)] \geq \delta \cdot (\frac{1}{4})^T \geq 2^{-\Theta(\sqrt{r} \log \frac{1}{\delta})}$.
We can combine those estimates and consider a single expectation
%
\[
  \E\Big[\mu(R \cap Q_1) - \frac{1}{\delta} \cdot \mu(R \cap Q_{-1})\Big] \geq 2^{-\Theta(\sqrt{r} \log \frac{1}{\delta})}
\]
We take any $R$ attaining this, then in particular we must have 
$\mu(R) \geq 2^{-\Theta(\sqrt{r} \log \frac{1}{\delta})}$ and $\mu(R \cap Q_{-1}) \leq \delta \cdot \mu(R)$.

\section{Remarks}
We want to conclude this paper with a couple of remarks:
\begin{itemize}
\item Instead of taking $T$ random Gaussians, one can also find 
an almost monochromatic rectangle using a \emph{single} Gaussian. 
Sample  $g \sim N^r(0,1)$ and define
\[
R := \{ x \in X : \left<\bar{u}_x,g\right> \geq s \} \times \{ y \in Y : \left<\bar{v}_y,g\right> \geq s\}.
\]
where  $s = \Theta(r^{1/4}\sqrt{\log r})$ is a suitable threshold. 
The rectangle $R$ will satisfy the same guarantee as before (up to constant factors). Geometrically, this approach might 
be more intuitive, as it means that one can take all vectors in a \emph{random cap} of the unit ball. 
\item The approach can also be used to get the discrepancy lower bound $\textrm{disc}(f) \geq \Omega(1/\sqrt{\rank(f)})$
as a corollary. Take any measure $\mu$ and assume that $\mu(Q_1) \geq 1/2$.
Then sample a single Gaussian $g \sim N^r(0,1)$ and let $R := \{ x \in X : \left<g,u_x\right> \geq 0 \} \times \{ y \in Y : \left<g,v_y\right> \geq 0\}$. 
Then 
\[
\E[\mu(R \cap Q_1) - \mu(R \cap Q_{-1})] \geq \frac{1}{2} \cdot \frac{1}{4} - \frac{1}{2} \cdot \Big(\frac{1}{4} - \frac{1}{7\sqrt{r}}\Big)= \frac{1}{14\sqrt{r}}.
\]
\item Note that John's theorem and also the bounds on $\|u_x\|_2,\|v_y\|_2$ are tight in general. However, it seems plausible that 
one can modify the hyperplane rounding in order to improve the bounds, possibly depending on the geometric arrangement of the vectors.
\item We do hesitate to call our proof ``new'', since all its ingredients have been contained already 
either in Lovett's paper~\cite{CommunicationBoundedByRootRank-Lovett-STOC2014} or in the paper of Linial et al.~\cite{Complexity-measures-of-sign-matrices-LMSS-Combinatorica2007}. For example the bound on the factorization norm~\cite{Complexity-measures-of-sign-matrices-LMSS-Combinatorica2007} is based on John's theorem; the hyperplane rounding is also used to prove
Grothendieck's inequality that is another ingredient of \cite{Complexity-measures-of-sign-matrices-LMSS-Combinatorica2007}.
Also \cite{CommunicationBoundedByRootRank-Lovett-STOC2014} used an amplification procedure as we did by sampling $T$ hyperplanes.
\item There are indeed rank-$r$ Boolean matrices that have $2^{\Omega(r)}$ many different rows and columns. The construction is due to Lov\'asz and Kotlov~\cite{RankOfGraphs-KotlovLovasz96}. 
An explicit construction is as follows: 
Take $8r$ disjoint symbols $A \dot{\cup} A' \dot{\cup} B \dot{\cup} B'$ with
$|A| = |A'| = |B| = |B'| = 2r$. 
Then define set families
\begin{eqnarray*}
  \mathcal{A} &:=& \{ a \subseteq A \cup B : |a \cap A| = r\textrm{ and }|a \cap B| = 1 \} 
 \cup \{ a \subseteq A' \cup B' : |a \cap A'| = r\textrm{ and }|a \cap B'| = 1\} \\
  \mathcal{B} &:=& \{ b \subseteq A' \cup B : |b \cap B| = r\textrm{ and }|b \cap A'| = 1\} \cup \{ b \subseteq A \cup B' : |b \cap B'| = r\textrm{ and }|B \cap A| = 1\}
\end{eqnarray*}
It is not difficult to check that for all $a \in \mathcal{A}$ and $b \in \mathcal{B}$ one has $|a \cap b| \in \{ 0,1\}$. Moreover for different tuples $a,a' \in \mathcal{A}$ there is
always a $b \in \mathcal{B}$ with $|a \cap b| = 0$ und $|a \cap b'| = 1$ and the reverse is true for different $b,b' \in \mathcal{B}$.
The matrix $M \in \{ \pm 1\}^{\mathcal{A} \times \mathcal{B}}$ defined by 
$M_{ab} = \left<(2 \cdot \bm{1}_a,-1),(2 \cdot \bm{1}_b,1)\right>$ has then rank at most $8r+1$
and $|\mathcal{A}| = |\mathcal{B}| = 2^{\Theta(r)}$ many different rows and columns.
\end{itemize}

\paragraph{Acknowledgments.}
The author is very grateful to Paul Beame, James Lee and Anup Rao for helpful discussions.

\bibliographystyle{alpha}
\bibliography{simpler-communication-proof}

\end{document}